\documentclass[11pt,a4paper]{article}

\usepackage{amssymb}
\usepackage{graphicx}
\usepackage{bm}

\usepackage{mathrsfs}
\usepackage{cite}
\usepackage{amsmath}

\usepackage{amsmath}
\usepackage{amsthm}

\usepackage{xcolor}

\usepackage{indentfirst}

\usepackage{array}
\usepackage{makecell}

\usepackage{textcomp}
\usepackage{float}

\usepackage{amsthm}

\begin{document}

\title{An Algorithm for Generating All Berglund–Hübsch-Type Calabi–Yau Orbifolds and Their Mirrors. Classification of Maximal Admissible Groups.}

\author{S.S.Aleshin  $\vphantom{\Big(}$
\medskip\\
${}${\small{\em Institute for Information Transmission Problems RAS,}}\\ {\small{ \em 127051, Moscow, Russia,}}\\ {\small{\em aless2001@mail.ru}}
}

\maketitle

\begin{abstract}
In this paper, we present an algorithm for constructing all Calabi–Yau orbifolds and their mirror pairs associated with an arbitrary Berglund–Hübsch-type model, and prove its correctness. In particular, the proposed method yields explicit generators for the groups defining the corresponding orbifolds. The method for classifying maximal admissible groups is proposed. We illustrate the application of the classification approach to an arbitrary Fermat-type Calabi–Yau manifold. In addition, as an illustrative example, we classify all maximal admissible groups of Fermat-type Calabi–Yau threefolds.

\end{abstract}


\section{Introduction}

The structure of the groups defining orbifolds and their mirrors encodes topological information about the corresponding spaces, including their Hodge numbers, Euler characteristics, and other topological invariants \cite{Roan,Vafa:1989xc,Aleshin:2026lep}. As is known, from a phenomenological point of view, it is important to be able to construct Calabi–Yau manifolds with a topology that determines the number of generations of quarks and leptons to be three. Some models with three generations are known \cite{Klemm:1992bx, Candelas:1989hd, Shishanin:2021thg}. Nevertheless, these lists may be far from complete, due to the possible existence of orbifolds associated with the models presented in these works, as well as with other models. The proposed algorithm enables us to systematically enumerate all orbifolds and their mirror pairs associated with each model and, consequently, to obtain a complete classification of those satisfying the prescribed properties. 

In \cite{Aleshin:2026sbv}, an algorithm for generating Berglund–Hübsch-type orbifolds of Calabi–Yau manifolds and their mirrors was developed. The approach developed in that work is based on both polynomial deformations and admissible groups. In the present work, a modified algorithm based exclusively on a group-theoretic approach is proposed. One of the advantages of the proposed method is the explicit construction of generators for the groups defining the orbifolds and their mirrors.

The paper is organized as follows. In Chapter 2, we formulate and prove the correctness of a modified algorithm for constructing orbifolds and their mirror pairs. We also present a method for obtaining generators of the groups defining the orbifolds.

Chapter 3 presents an algorithm for obtaining representatives of all conjugacy classes of elements of the quotient group $H/J_A$, where $H$ is an arbitrary subgroup satisfying $J_A \leq H \leq G_{\mathrm{adm}}^{\mathrm{max}}$, and $J_A$ is the symmetry group of the weighted projective space.

In Chapter 4, we present a method for classifying all maximally admissible groups and their quotients by the symmetry group of the weighted projective space for Fermat, Loop, Chain polynomials, and their combinations. Namely, we obtain a decomposition of the maximal admissible group into invariant factors (the Frobenius representation of the group), as well as a method for passing from elements of the group in such a decomposition to the maximal admissible group defined by Krawitz. This decomposition makes it possible to substantially simplify the first step of the algorithm described in Chapter 2.

In Chapter 5, we illustrate the application of the classification approach developed in Chapter 4 to an arbitrary Fermat-type Calabi–Yau manifold. In addition, as an illustrative example, we classify all maximal admissible groups of Fermat-type Calabi–Yau threefolds in the Appendix.


We consider the class of Calabi–Yau manifolds of Berglund–Hübsch type. These manifolds are defined as hypersurfaces in weighted projective spaces by polynomials satisfying the following conditions:

\begin{eqnarray}\label{P}
&&\mathbb{P}_{{\bf k}}^{n-1} =
\Big\{(x_1,\ldots,x_n)\in \mathbb{C}^n\setminus \{0\}
\,\,\Big|\,\,
x_i \sim \lambda^{k_i}x_i,\quad
\forall\, \lambda\in \mathbb{C}
\Big\}\\
&&W_0 := \sum_{i=1}^n\prod_{j=1}^n x_j^{A_{ij}}=0,\quad A\in \mathrm{Mat}_{n\times n}(\mathbb{Z}),\quad \det A\neq 0.\\
&&\sum_{j=1}^n A_{ij} k_j = d \label{quasi_homogeneous_polynomial}\\
&&d = \sum_{i=1}^n k_i\label{c1_0}.
\end{eqnarray}

\noindent
where $W_0$ has no critical points outside the origin and the last condition is equivalent to the vanishing of the first Chern class of the hypersurface. It is known \cite{Kreuzer:1992bi} that $W_0$ is invertible and non-degenerate if and only if the polynomial $W_0$ can be represented as a sum of atomic-type polynomials

\begin{eqnarray}
&&W_{\text{Fermat}}=\sum_{i=1}^n x_i^{A_{ii}}\nonumber\\
&&W_{\text{Chain}}=\sum_{i=1}^{n-1} x_i^{A_{ii}}x_{i+1}
+ x_n^{A_{nn}}\\
&&W_{\text{Loop}}=\sum_{i=1}^{n-1} x_i^{A_{ii}}x_{i+1}
+ x_n^{A_{nn}}x_1. \nonumber
\end{eqnarray}

\noindent
For this class of Calabi–Yau manifolds, Berglund and Hübsch \cite{Berglund:1991pp} proposed a mirror construction in which the mirror manifolds arise as particular orbifolds of projective spaces associated with the corresponding $A^T$ polynomials. This construction was subsequently generalized by Krawitz \cite{Krawitz}.

We proceed to the construction of Calabi–Yau orbifolds. For this purpose, we define the maximal group of diagonal symmetries.

\begin{eqnarray}\label{Aut_0}
&&Aut(A):=\Big\{(\alpha_1,...,\alpha_n)\in(\mathbb{C^*})^n\,\,\Big|\,\, W_0(\alpha_1\cdot x_1,...,\alpha_n\cdot x_n)= W_0(x_1,...,x_n)\Big\}
\end{eqnarray}

\noindent
It can be shown that the invariance condition for the polynomial $W_0$ in (\ref{Aut_0}) holds for any polynomial of one of the atomic types (Fermat, Loop, Chain), as well as for any polynomial that is a combination of such types, it is true that 
$|\alpha_i|=1$.  Then $\alpha_i = e^{2\pi i \theta_i}$, and to satisfy the condition of invariance of the polynomial $W_0$ in (\ref{Aut_0}), it is required that

\begin{eqnarray}
\sum_{j=1}^n A_{ij}\theta_j \in \mathbb{Z}^n.
\end{eqnarray}

\noindent
Since the matrix $A$ is invertible, the generators of the symmetry group of the polynomial $W_0$ are given by

\begin{eqnarray}\label{gen_B}
\Big(e^{2\pi i B_{1j}},..., e^{2\pi i B_{nj}}\Big),\quad j\in\overline{1,n}
\end{eqnarray}

\noindent
where $B_{ij}:= (A^{-1})_{ij}$, then, for an arbitrary $(\alpha_1,...,\alpha_n)\in Aut(A)$ there exist integers $m_j$,  $j\in\overline{1,5}$, such that

\begin{eqnarray}\label{alpha_i}
\alpha_i = exp\Big (2\pi i \sum_{j=1}^n B_{ij}m_j\Big),\,\quad \forall\,i\in \overline{1,n}
\end{eqnarray}

\noindent
From the definition (\ref{Aut_0}) it follows that the element $(\alpha_1,\ldots,\alpha_n)\in Aut(A)$ acts as

\begin{eqnarray}
x_i \mapsto \alpha_i\cdot x_i = \exp\Big (2\pi i \sum_{j=1}^nB_{ij}m_j\Big)x_i
\end{eqnarray}

\noindent
Let us consider in $Aut(A)$ the subgroup that preserves $\prod_{i=1}^nx_i$
\begin{eqnarray}\label{SL_pred_final}
G_{\text{adm}}^{\text{max}}=\Big\{p\in Aut(A)\,\,\Big |\,\, p\cdot \prod_{i=1}^nx_i=\prod_{i=1}^nx_i \Big\}
\end{eqnarray}

\noindent
This group is referred to as the maximal admissible group. It follows from (\ref{alpha_i}) and the definition (\ref{SL_pred_final}) that the maximal admissible group can be represented as

\begin{eqnarray}\label{SL_pred_final_new}
G_{\text{adm}}^{\text{max}} = \Big\{\left(e^{2\pi i\sum_{j=1}^nB_{1j}m_j},\ldots,e^{2\pi i\sum_{j=1}^nB_{nj}m_j}\right)\in (\mathbb{C}^*)^n \,\,\Big |\,\,\sum_{i,j=1}^n B_{ij}m_j\in\mathbb{Z}\Big\}.
\end{eqnarray}

\noindent
The group $G_{\mathrm{adm}}^{\mathrm{max}}$ contains the quantum group, i.e., the symmetry group of the weighted projective space, as a subgroup.

\begin{eqnarray}
J_A := \left\langle \Big(e^{\frac{2\pi i}{d}k_1},\dots,e^{\frac{2\pi i}{d}k_n}\Big)\right\rangle
\end{eqnarray}

\noindent
that is, the group generated by the element $(e^{\frac{2\pi i}{d}k_1},\dots,e^{\frac{2\pi i}{d}k_n})\in G_{\text{adm}}^{\text{max}}$. 
For the orbifold construction, we will use admissible groups 
$G$ (subgroups of $G_{\text{adm}}^{\text{max}}$) for which
$J_A\subseteq G\subseteq G_{\text{adm}}^{\text{max}}$ is satisfied.

\section{Algorithm for finding Berglund--Hubsch type Calabi--Yau orbifolds and their mirror pairs}
\hspace*{\parindent}

In this section, we propose a modification of the algorithm introduced in \cite{Aleshin:2026sbv} for constructing Berglund--Hubsch-type orbifolds and their mirror pairs. Unlike the previous method, this algorithm is based solely on a group approach and does not use deformations. Also, we will propose a method for finding generators of the groups defining the orbifolds.



In the first step we find all  elements of the maximal admissible groups $G_{\mathrm{adm}}^{\max}$ and $G_{\mathrm{adm}}^{\max, T}$. Further, the index "T" will mean that the group is a mirror group in relation to the one without this index. To compute the elements of the group $G_{\mathrm{adm}}^{\max, T}$, we must take the polynomial determined by the transposed matrix $A^T$. Chapter 4 provides an algebraic description of the structure of the maximal admissible group. Here, we give a computational approach to the problem. To determine all elements of the group $G_{\mathrm{adm}}^{\max}$, it suffices, by (\ref{SL_pred_final_new}), to find the solutions of

 \begin{eqnarray}\label{prod_x_i}
\sum_{i,j=1}^n B_{ij}m_j\in\mathbb{Z}.
\end{eqnarray}

\noindent
We shall determine the solution by considering all admissible values of $m_j$. To this end, we derive the constraints imposed on $m_j$.
Observe that $B_{ij}\in\mathbb{Q}$, thus $B_{ij}=c_{ij}/d_{ij}$ for some $c_{ij}\in\mathbb{Z},$ $d_{ij}\in\mathbb{Z}\setminus\{0\}$.
To determine the maximal possible value of an arbitrary fixed $m_p$, we consider

\begin{equation}
\sum_{\substack{i,j=1\\ j\neq p}}^n B_{ij}m_j
+
\left(\sum_{i=1}^n \frac{c_{ip}}{d_{ip}}\right)m_p
\in \mathbb{Z}.
\end{equation}

\noindent
It follows that

\begin{equation}
0 \leq m_p < \operatorname{lcm}(d_{1p},\ldots,d_{np}),
\qquad p\in\overline{1,n}.
\end{equation}










\noindent
At the second stage,

\begin{enumerate}
    \item For each element $g=(e^{2\pi i g_1},\ldots,e^{2\pi i g_n}) \in G_{\mathrm{adm}}^{\max}$, we determine all corresponding mirror elements $q$ of the group $G_{\mathrm{adm}}^{\max,T}$ satisfying the Krawitz mirror image condition

\begin{eqnarray}\label{mirror_condition}
&R_g=\left\{q \in G_{\mathrm{adm}}^{\max, T}\;\middle|\;\sum_{i,j=1}^{n}g_i A_{ij} q_j\in \mathbb{Z}
\right\}.
\end{eqnarray}
Note that $R_g$ is a group. Enumerate all unique groups $R_g$. We obtain the set $\left\{R_i\right\}$.

\item For each $R_i$, define the associated set
\begin{eqnarray}
G_i=\left\{g\in G_{\mathrm{adm}}^{\max}\mid R_g=R_i\right\}.
\end{eqnarray}
This yields pairs $(G_i, R_i)$, where both $G_i$ and $R_i$ consist of elements of the groups $G_{\mathrm{adm}}^{\max}$ and $G_{\mathrm{adm}}^{\max, T}$, respectively.
\end{enumerate}

\newpage

\newtheorem{proposition}{Proposition}
\begin{proposition}\label{Ri_is_cyclic_group}
\leavevmode\par\vspace{0.5\baselineskip}
The group $R_g^T$ is cyclic and is generated by an arbitrary element of the set $G_g$.

\end{proposition}

\vspace{0.0001cm}

\renewcommand{\proofname}{Proof}

\begin{proof}
\leavevmode\par\vspace{0.5\baselineskip}
Let $Q := (R_g)^T$. Consider the cyclic group $\langle t\rangle$ generated by an arbitrary element $t\in G_g$. Then

\begin{eqnarray}
\langle t\rangle ^T
&=&
\left\{
h\in G_{\mathrm{adm}}^{\max,T} \;\middle|\;
\sum_{i,j=1}^{n} f_iA_{ij}h_j\in\mathbb{Z},
\ \forall\, f\in\langle t\rangle
\right\}
\nonumber\\[1ex]
&=&
\left\{
h\in G_{\mathrm{adm}}^{\max,T} \;\middle|\;
\sum_{i,j=1}^n t_iA_{ij}h_j\in\mathbb{Z}
\right\}=R_t=R_g,
\end{eqnarray}

\noindent 
where the last equality follows from $t\in G_g$. By Lemma 3.3 of Kravitz \cite{Krawitz}, for group $G$ we have $G = \left(G^{T}\right)^{T}$. Hence,

\begin{eqnarray}
Q=(Q^T)^T
=\bigl[(R_g^T)^T\bigr]^T= R_g^T
=(\langle t\rangle^T)^T=\langle t\rangle.
\end{eqnarray}

\end{proof}

\noindent 
It follows that any element of the set $G_i$ may be chosen as a generator of a group $R_i^T$, and this group will be mirror to the group $R_i$. We are now ready to describe any subgroup defining an orbifold in terms of $R_i$ - groups by means of the following statement.

\begin{proposition}\label{Proof_of_the_algorithm}
\phantom{test}

\begin{enumerate}
\renewcommand{\labelenumi}{(\arabic{enumi})}
\item For any subgroup $H$, $J_A\subset H\subset G_{\text{adm}}^{\text{max}}$ with generators $h_{i_j}\in G_i$, $i_j\in I_H$ which we will denote as $H=\langle g_0, h_{i_j}\mid i_j\in I_H\rangle$ and $J_A=\langle g_0 \rangle$, is satisfied

\begin{eqnarray}\label{group_description_in_terms_of_R_sets_intersections}
H^T = \bigcap_{i\in \overline{I_H}}R_{h_i} = \bigcap _{i\in \overline{I_H}} R_i
\end{eqnarray}

where
$g_0=\left(e^{2\pi i \frac{k_1}{d}},\dots,e^{2\pi i \frac{k_n}{d}}\right)$ is the symmetry generator of the weighted projective space and $h_i$ is any fixed element from $G_{i}$, $\overline{I_H}$ is the number of unique indices $i$ included in the set $I_H$.

\item Conversely, for any intersection $\bigcap_{i\in I} R_i$ there exists a subgroup $H$ such that

\begin{eqnarray}\label{intersection_as_a_group}
\bigcap_{i\in I} R_i = H^T.
\end{eqnarray}

 To find the generators of $H$, it is enough to select one arbitrary element $g_i$ from each $G_i$ $i\in I$, then their independent subset (together with $g_0$) will be the set of generators of group~$H$

\begin{equation}\label{constructive_production_of_group_H_generators}
\forall (g_i)_{i \in I} \in \prod_{i \in I} G_i,\qquad
H = \langle g_0, g_i \mid i \in I \rangle.
\end{equation}

\end{enumerate}

\end{proposition}

\vspace{0.0001cm}

\renewcommand{\proofname}{Proof}

\begin{proof}
\mbox{}

\begin{description}
\item[{\bf (1)}]

\begin{eqnarray}
 H^T =\bigcap_{h\in H}R_h=R_{\bigcup_{i_j\in I_H}h_{i_j}}=\bigcap_{i_j\in I_H}R_{h_{i_j}}=\bigcap_{i\in \overline{I_H}}R_{h_{i}}=\bigcap _{i\in \overline{I_H}} R_i
\end{eqnarray}

\noindent
where $R_A$ is a set of elements from $G_{\mathrm{adm}}^{\max, T}$ such that their pairing with any of the elements from set A satisfies condition (\ref{mirror_condition}).

\item[{\bf (2)}]

From (\ref{constructive_production_of_group_H_generators}) and (\ref{group_description_in_terms_of_R_sets_intersections}) it follows directly that there is a group $H$ such that (\ref{intersection_as_a_group}) is satisfied.



\end{description}
\end{proof}

Proposition~\ref{Proof_of_the_algorithm} establishes a bijective correspondence between orbifolds and intersections of the groups $R_i$. Moreover, the above proposition provides an algorithm for obtaining generators of the group $H$ defining the orbifold. More precisely, by part~(1), each such group $H$ corresponds to an intersection of groups $R_i$ by relation (\ref{group_description_in_terms_of_R_sets_intersections}), while part~(2) describes how to determine generators for a given intersection group.

At the third step, we determine all possible intersections of the groups $R_i$. To determine all types of intersections, we proceed inductively from $(N-1)$-fold intersections to $N$-fold intersections, starting from the single intersections $R_i$, $i \in \overline{1,N}$. To obtain new intersection types among the $N$-fold intersections, we take each intersection type arising from the $(N-1)$-fold intersections and intersect it successively with all sets $R_i$, $i \in \overline{1,N}$, recording the newly obtained intersection types. If, starting from some step, no new intersection types appear when passing to the next multiplicity, the procedure terminates.

From Proposition~\ref{Proof_of_the_algorithm} it follows that to determine a generating set for the mirror group corresponding to the group given by the intersection $\cap_{i\in I} R_i$, it suffices to choose one arbitrary element from each set $G_i$, where $i\in I$. The resulting collection of elements forms a generating set for $(\cap_{i\in I} R_i)^T$:

\begin{equation}
\left( \bigcap_{i \in I} R_i \right)^{T}
=
\left\langle g_i \in G_i \mid i \in I \right\rangle .
\end{equation}


In the next step we repeat Steps 2 and 3 for the elements of $G_{\mathrm{adm}}^{\max,T}$.

As a result, we obtain two lists of groups defining all orbifolds. Using the mirror condition (\ref{mirror_condition}) we pair the resulting groups into mirror pairs under the requirement that the mirror group is of maximal order. This produces the complete collection of mirror pairs.

\section{Algorithm for computing $H/J_A$ for an arbitrary Berglund–Hubsch-type model}
\hspace*{\parindent}

Consider an arbitrary subgroup \(H \leq G_{\mathrm{adm}}^{\max}\).
The group $H$ can be quotiented by the symmetry group of the weighted projective space $J_A$. 
To factorize the group $H$ by the symmetry group of the weighted projective space $J_A$, we proceed as follows. We consider the elements $f_i \in H$ one by one. If the element $f_i$ does not belong to the set

\begin{equation}
C_{i-1}:=\left\{f_j g_0^p\;\middle|\;j \in \overline{1,i-1},\;p \in\overline{0,d}
\right\},
\end{equation}

\noindent 
then we record $f_i$ as a representative of a coset in $H/J_A$ and define the set

\begin{equation}
C_i:=C_{i-1}\cup\left\{f_i g_0^p\;\middle|\;p \in \overline{0,d}\right\}.
\end{equation}

\noindent 
We continue this procedure until all elements of $G_{\mathrm{adm}}^{\max}$ have been exhausted. At the end of the algorithm, we obtain a complete set of representatives of the cosets in $H/J_A$.

\section{Classification of maximal admissible groups}
\hspace*{\parindent}


Let us consider condition (\ref{prod_x_i}). In the notation of Chapter 2, this condition can be written as

\begin{eqnarray}\label{Diophantine_eq_G_adm_max}
&&\sum_{i,j=1}^n B_{ij}m_j=\sum_{i=1}^n\frac{1}{\operatorname{lcm}_i}\sum_{j=1}^n F_{ij}m_j=\nonumber\\
&&=\frac{1}{L}\left(\sum_{i=1}^n\frac{L}{\operatorname{lcm}_i}F_{i1}\right)m_1+\ldots + \frac{1}{L}\left(\sum_{i=1}^n\frac{L}{\operatorname{lcm}_i}F_{in}\right)m_n \in\mathbb{Z},\\
&&F_{ij}:=c_{ij}\frac{\operatorname{lcm}_{i}}{d_{ij}}\in\mathbb{Z} \quad \forall\,i\in \overline{1,n},\nonumber
\end{eqnarray}

\noindent
where least common multiple $\operatorname{lcm}_{i}=\operatorname{lcm}(d_{i1},\ldots d_{in})$ and $L$ is the least common multiple $L=\operatorname{lcm}(\operatorname{lcm}_1,\ldots,\operatorname{lcm}_n)$. Since $\operatorname{lcm}_i \mid L$ and $F_{ij}\in\mathbb{Z}$, then

\begin{eqnarray}\label{G_adm_max_new}
&&G_{\text{adm}}^{\text{max}} \cong \Big\{(m_1,...,m_n)\in \mathbb{Z}^n \,\,\Big |\,\,\frac{1}{L}\sum_{j=1}^n G_jm_j\in\mathbb{Z}\Big\}\Big /T_1\mathbb{Z}\times...\times T_n\mathbb{Z},\\
&&G_j:=\sum_{i=1}^n
\frac{L}{\operatorname{lcm}_i}F_{ij}, \quad T_j:=\frac{\operatorname{lcm}(G_j,L)}{G_j},\quad j\in\overline{1,n}.
\end{eqnarray}

\noindent
The condition in group (\ref{G_adm_max_new})
\begin{eqnarray}
\sum_{j=1}^n G_jm_j \equiv 0 \pmod{L}
\end{eqnarray}

\noindent
can be represented  as a homogeneous Diophantine equation in the variables $m,\, m_i,$ $i\in\overline{1,n}$
\begin{eqnarray}\label{fin_Diophantine_eq_G_adm_max_1}
\sum_{j=1}^n G_jm_j  + Lm = 0, \,\, m_i,m\in \mathbb{Z}
\end{eqnarray}

\noindent
Next, we find a fundamental system of solutions to equation (\ref{fin_Diophantine_eq_G_adm_max_1}) 

\begin{eqnarray}\label{fundamental_system_G_adm_max}
v^j = \begin{pmatrix}
m_1^j \\
\vdots \\
m_n^j
\end{pmatrix},\,\, j\in\overline{1,n}.
\end{eqnarray}

\noindent
where $j\in\overline{1,n}$ denotes the index of the fundamental solution.
It follows that the maximal admissible group can be represented as a quotient of a free abelian group.

\begin{eqnarray}\label{lattice_before_transition_new}
G_{\text{adm}}^{\text{max}} \cong \bigoplus_{i=1}^n \mathbb{Z}v^i \Big/ L_{R}
\end{eqnarray}

\noindent
where the submodule $L_{R}=span_{\mathbb{Z}}(R)\subset \mathbb{Z}^n$ is the $\mathbb{Z}$-linear span of the column vectors  (generating relations) $R\subset \mathbb{Z}^n$.

\begin{eqnarray}\label{}
R = \begin{pmatrix}
T_{1} & 0      & \cdots & 0 \\
0      & T_{2} & \cdots & 0 \\
\vdots & \vdots & \ddots & \vdots \\
0      & 0      & \cdots & T_{n}
\end{pmatrix}
\end{eqnarray}

We now determine which element of the group in (\ref{G_adm_max_new}) induces a symmetry of the weighted projective space. Since a generator of the symmetry group of the weighted projective space, viewed as a subgroup of (\ref{SL_pred_final_new}), is of the form $\left(e^{2\pi i \frac{k_1}{d}},\ldots,e^{2\pi i \frac{k_n}{d}}\right),$
and since

\begin{eqnarray}
\sum_{j=1}^n A_{ij}k_j=d\mathbf{1}_i \Rightarrow \sum_{i,j=1}^n B_{fi}A_{ij}k_j = \sum_{i=1}^n B_{fi}d\mathbf{1}_i \Rightarrow \sum_{j=1}^n\delta_{fj}k_j=d\sum_{i=1}^nB_{fi}\nonumber
\end{eqnarray}

\begin{eqnarray}\label{generator_J_A_new}
 \frac{k_f}{d} = \sum_{i=1}^n B_{fi},
\end{eqnarray}

\noindent
where $\mathbf{1}$ denotes the all-ones column vector. It follows from (\ref{generator_J_A_new}) that the generator of the symmetry group of the weighted projective space, viewed as a subgroup of (\ref{G_adm_max_new}), is given by the all-ones column vector. Taking into account the symmetry of the weighted projective space, we obtain

\begin{eqnarray}
\begin{pmatrix}
1 \\
\vdots \\
1
\end{pmatrix} \sim \begin{pmatrix}
0 \\
\vdots \\
0
\end{pmatrix}\nonumber
\end{eqnarray}

\noindent
then the matrix of relations will take the form

\begin{eqnarray}\label{}
Q = \begin{pmatrix}
T_{1} & 0      & \cdots & 0 & 1 \\
0      & T_{2} & \cdots & 0 & 1 \\
\vdots & \vdots & \ddots & \vdots & \vdots\\
0      & 0      & \cdots & T_{n} & 1
\end{pmatrix}
\end{eqnarray}

\noindent
then

\begin{eqnarray}\label{lattice_proj_before_transition_new}
G_{\text{adm}}^{\text{max}}/J_A \cong \bigoplus_{i=1}^n \mathbb{Z}v^i \Big/ L_{Q}
\end{eqnarray}

\noindent
Since any column vectors in $R$ and $Q$ is a solution of (\ref{fin_Diophantine_eq_G_adm_max_1}) then we can find the coordinates of the column vectors of $R$ and $Q$ in the basis of the fundamental system of solutions (\ref{fundamental_system_G_adm_max}), that is, the matrices $R_v$ and $Q_v$

\begin{eqnarray}\label{}
\sum_{i=1}^n v_{j}{}^{i}(R_v)_{ir} = R_{jr},\quad \sum_{i=1}^n v_{j}{}^{i}(Q_v)_{ir} = Q_{jr},
\end{eqnarray}

\noindent
where the first and second indices of the matrices indicate the row and column, respectively. Then relations (\ref{lattice_before_transition_new}) and (\ref{lattice_proj_before_transition_new}) take the form

\begin{eqnarray}
G_{\text{adm}}^{\text{max}} \cong \mathbb{Z}^n \Big/ L_{R_v},\quad G_{\text{adm}}^{\text{max}}/J_A \cong \mathbb{Z}^n \Big/ L_{Q_v}.
\end{eqnarray}

\noindent
as $n = \operatorname{rk}(R_v) = \operatorname{rk}(L_{R_v})$ and $n=\operatorname{rk}(Q_v)=\operatorname{rk}(L_{Q_v})$ then the ranks of the free summands in the direct sum decomposition of the groups into cyclic groups are $\operatorname{rk}\!\left(
G_{\mathrm{adm}}^{\max}/\operatorname{Tors}(G_{\mathrm{adm}}^{\max})\right)= n - \operatorname{rk}(R_v) = 0$ and $\operatorname{rk}\!\left(
\left( G_{\mathrm{adm}}^{\max}/J_A\right)/\operatorname{Tors}\!\left(G_{\mathrm{adm}}^{\max}/J_A\right)\right)= n - \operatorname{rk}(Q_v) = 0$. Hence, there are no free summands in the decomposition of $G_{\text{adm}}^{\text{max}}$ and $G_{\text{adm}}^{\text{max}}/J_A $.

Thus, the problem reduces to studying the structure of the $\mathbb{Z}$-module $\mathbb{Z}^n / L_P$, where $P=R_v$ or $P=Q_v$. We can now apply the theorem on invariant factors:

\newtheorem{theorem}{Theorem}
\begin{theorem}\label{invariant_factors_0}
\leavevmode\par\vspace{0.5\baselineskip}
Let $M$ be a free module of rank $n$ over a principal ideal domain $K$, and let $N\subset M$ be a submodule. Then there exists a basis $(e_1,\ldots,e_n)$ of $M$ such that $\lambda_1 e_1,\ldots,\lambda_m e_m$
is a basis of $N$, where $m=\operatorname{rk}N$, and the elements \(\lambda_1,\ldots,\lambda_m\) are the invariant factors of $N$ in $M$, satisfying $\lambda_i\mid\lambda_j$ for $i<j$. (We shall refer to the bases $\{e_i\}$ and $\{\lambda_i e_i\}$ as compatible bases.)
\end{theorem}

\noindent
By Theorem \ref{invariant_factors_0}, applied to the inclusion $L_P\subset \mathbb{Z}^n$ of $\mathbb{Z}$-modules, the quotient group $\mathbb{Z}^n/L_P$ admits a decomposition into invariant factors (Frobenius representation of a group)

\begin{eqnarray}\label{P_and_group_decomposition_new}
\mathbb{Z}^n\Big /L_P \cong \mathbb{Z}/(\lambda_1)\oplus...\oplus \mathbb{Z}(\lambda_r).
\end{eqnarray}

\noindent
where $r=rk(L_P)=rk(P)=n$ and $\mathbb{Z}/(1)=\{0\}$. To determine the invariant factors, we apply the Smith normal form algorithm to the relation matrix $P$

\begin{theorem}\label{Smith_normal_form_new}
\leavevmode\par\vspace{0.5\baselineskip}
\begin{enumerate}
\item For any matrix $P\in Mat_{n\times m}(K)$ over a principal ideal domain $K$, there exist invertible matrices $U_P\in Mat_{n\times n}(K)$ 
and $F_P\in Mat_{m\times m}(K)$ such that

\begin{eqnarray}
D_P = U_PPF_P
\end{eqnarray}

\noindent
where $(D_P)_{ij}=0$ for $i\neq j$, and $(D_P)_{ii}\mid(D_P)_{jj}$ for $i<j$.

\item If the sublattice $L_P \subset \mathbb{Z}^n$ is generated by the columns of the matrix $P$, then the compatible bases of $\mathbb{Z}^n$ and $L_P$ are given, respectively, by the columns of the matrix $U_P^{-1}$ and the vectors $(D_P)_{ii}[U_P^{-1}]_{*i}$, $(D_P)_{ii}\neq 0$ ($[U_P^{-1}]_{*i}$ denotes the $i$-th column of the matrix $U_P^{-1}$).

\end{enumerate}
\end{theorem}

\noindent
It follows from Theorem \ref{Smith_normal_form_new} that the decomposition (\ref{P_and_group_decomposition_new}) can be obtained by reducing the matrix $P$ to its Smith normal form, since $\lambda_i=(D_P)_{ii}$, $(D_P)_{ii}\neq 0$ (for further details on Theorems 1 and 2, see, for example, \cite{Gorodentsev}). It follows from the above that there exist isomorphisms $\varphi_1$ and~$\varphi_2$.

\begin{eqnarray}\label{Smith1_new}
&&\varphi_1:\; G_{\text{adm}}^{\text{max}} \xrightarrow{\cong} \bigoplus_{i=1}^n \mathbb{Z}/(\mu_i),\qquad\,\,\, \mu_l\mid\mu_p\,\,\text{for}\,\, l<p\\ 
\label{Smith2_new}
&&\varphi_2:\;G_{\text{adm}}^{\text{max}}/J_A \xrightarrow{\cong} \bigoplus_{j=1}^n \mathbb{Z}/(\nu_j),\quad  \nu_s\mid\nu_c\,\,\text{for}\,\, s<c.
\end{eqnarray}

\noindent
where $\mu_i=(D_{R_v})_{ii}$, $\nu_j=(D_{Q_v})_{jj}$ ($(D_{R_v})_{ii}\neq 0$, $(D_{Q_v})_{jj}\neq 0$) and  some of the direct summands in the decompositions may be trivial.

Having obtained the invariant factor decompositions of the groups $G_{\mathrm{adm}}^{\max}$ and $G_{\mathrm{adm}}^{\max}/J_A$, we can now establish a correspondence between the elements of (\ref{G_adm_max_new}) and the elements of the groups on the right-hand sides of the isomorphisms (\ref{Smith1_new}) and (\ref{Smith2_new}). For definiteness, assume that \(\mu_1=\cdots=\mu_q=1\) and \(\mu_{q+1}\mid\cdots\mid\mu_n\) in (\ref{Smith1_new}), and that \(\nu_1=\cdots=\nu_p=1\) and \(\nu_{p+1}\mid\cdots\mid\nu_n\) in~(\ref{Smith2_new}).

Let $g=\bigl([y_{q+1}]_{\mu_{q+1}},\ldots,[y_n]_{\mu_n}\bigr)
\in \operatorname{Im}(\varphi_1)$, then it follows from the construction that

\begin{eqnarray}\label{varphi_1_-1}
\varphi_1^{-1}(g)=\left[VU_{R_v}^{-1}y\right]\in G_{\mathrm{adm}}^{\max}
\end{eqnarray}

\noindent
where $y=(0,\ldots,0,y_{q+1},\ldots,y_n)^T$ is $n$-dimensional column vector and the square brackets denote an equivalence class
\begin{eqnarray}
[a]=[b] \quad\Longleftrightarrow\quad a-b\in\mathbb{Z}\begin{pmatrix}
T_1 \\
\vdots \\
T_n
\end{pmatrix}.
\end{eqnarray}

If $h=([h_{p+1}]_{\nu_{p+1}},\ldots,[h_n]_{\nu_n})\in \operatorname{Im}(\varphi_2)$, consider the quotient homomorphism
\begin{equation}\label{s}
\pi:\; G_{\mathrm{adm}}^{\max} \to G_{\mathrm{adm}}^{\max}/J_A.
\end{equation}

\noindent
Then
\begin{equation}\label{varphi_2_circ_pi_-1}
(\varphi_2\circ\pi)^{-1}(h)=\left\{
\left.
[VU_{Q_v}^{-1}z + s\mathbf{1}]\in G_{\mathrm{adm}}^{\max}
\;\right|\;
0\leq s < \operatorname{lcm}(T_1,\ldots,T_n),\;s\in\mathbb{Z}
\right\},
\end{equation}

\noindent
where $\mathbf1=(1,\ldots,1)^T$ and $z=(0,\ldots,0,h_{p+1},\ldots,h_n)^T$ are $n$-dimensional column vectors.

Conversely, let $f=[(w_1,\ldots,w_n)]\in G_{\mathrm{adm}}^{\max}$, where \(G_{\mathrm{adm}}^{\max}\) is defined in (\ref{G_adm_max_new}) then
\begin{align}
c &:= \left(VU_P^{-1}\right)^{-1}w, \\ \label{varphi1_f}
\varphi_1(f)
&= \bigl([c_{q+1}]_{\mu_{q+1}},\ldots,[c_n]_{\mu_n}\bigr), \\\label{varphi2_pi_f}
(\varphi_2\circ\pi)(f)
&= \bigl([c_{p+1}]_{\nu_{p+1}},\ldots,[c_n]_{\nu_n}\bigr),
\end{align}

\noindent
where $w:=(w_1,\ldots,w_n)$, and $P=R_v$ or $P=Q_v$, depending on whether we consider (\ref{varphi1_f}) or (\ref{varphi2_pi_f}), respectively.

Thus, in order to construct the groups $G_{\mathrm{adm}}^{\max}$ and $G_{\mathrm{adm}}^{\max}/J_A$, it suffices to determine a fundamental system of solutions \(V\) and reduce the matrices $R_v$ and $Q_v$ to Smith normal form. Using (\ref{varphi_1_-1}) and (\ref{varphi_2_circ_pi_-1}), we will be able to recover the actions of the elements of $\operatorname{Im}(\varphi_i)$, $i=1,2$, on the polynomial. This approach considerably simplifies the computations required for the first step of the algorithm proposed in Chapter 2.

\section{Example: Classification of maximal admissible groups for Fermat-type Calabi-Yau}
\hspace*{\parindent}

Consider a Fermat-type polynomial

\begin{eqnarray}
W_0 = \sum_{i=1}^m x_i^{A_{ii}}
\end{eqnarray}

\noindent
From (\ref{gen_B}) it follows that the generators of the group 
$Aut(A)$ for Fermat-type models have the form

\begin{eqnarray}
 \Big(e^{2\pi i \frac{1}{A_{11}}\delta_{1j}},..., e^{2\pi i \frac{1}{A_{mm}}\delta_{mj}}\Big),\quad j\in\overline{1,m}
\end{eqnarray}

\noindent
Then the maximal admissible group can be represented as

\begin{eqnarray}\label{SL_pred_final_new_new}
&&G_{\text{adm}}^{\text{max}} =\Big\{\prod_{j=1}^m\Big(e^{2\pi i \frac{1}{A_{11}}\delta_{1j}},...,e^{2\pi i \frac{1}{A_{mm}}\delta_{mj}}\Big)^{n_j}\in(\mathbb{C}^*)^m\,\,\Big |\,\, \sum_{j=1}^m \frac{1}{A_{jj}}n_j\in\mathbb{Z}\Big\}\cong \nonumber\\
&& \cong \Big\{(n_1,...,n_m)\in \mathbb{Z}^m \,\,\Big |\,\,\sum_{j=1}^m \frac{1}{A_{jj}}n_j\in\mathbb{Z}\Big\}\Big /A_{11}\mathbb{Z}\times...\times A_{mm}\mathbb{Z}.
\end{eqnarray}

\noindent
The condition on the group in (\ref{SL_pred_final_new_new}) can be expressed as

\begin{eqnarray}\label{Diophantine_equation_0}
\sum_{i=1}^m a_in_i  \equiv 0 \pmod{b}
\end{eqnarray}

\noindent
where $b=\operatorname{lcm}(A_{11},...,A_{mm})$ and $a_i=b/A_{ii}$. Condition (\ref{Diophantine_equation_0}) can be represented  as a homogeneous Diophantine equation in the variables $n,\, n_i,$ $i\in\overline{1,m}$

\begin{eqnarray}\label{Diophantine_equation_1}
\sum_{i=1}^m a_in_i  + bn = 0, \,\, n_i,n\in \mathbb{Z}
\end{eqnarray}

\noindent
Next, we find a fundamental system of solutions to equation (\ref{Diophantine_equation_1}) 

\begin{eqnarray}\label{fundamental_system}
v^j = \begin{pmatrix}
n_1^j \\
\vdots \\
n_m^j
\end{pmatrix},\,\, j\in\overline{1,m}.
\end{eqnarray}

\noindent
Then the maximal admissible group can be represented as a quotient of a free abelian group.

\begin{eqnarray}\label{lattice_before_transition}
G_{\text{adm}}^{\text{max}} \cong \bigoplus_{i=1}^m \mathbb{Z}v^i \Big/ L_{R}
\end{eqnarray}

\noindent
where the submodule $L_{R}=span_{\mathbb{Z}}(R)\subset \mathbb{Z}^m$ is the $\mathbb{Z}$-linear span of the column vectors  (generating relations) $R\subset \mathbb{Z}^m$.

\begin{eqnarray}\label{}
R = \begin{pmatrix}
A_{11} & 0      & \cdots & 0 \\
0      & A_{22} & \cdots & 0 \\
\vdots & \vdots & \ddots & \vdots \\
0      & 0      & \cdots & A_{mm}
\end{pmatrix}
\end{eqnarray}

If, in addition, we want to take into account the symmetry of the weighted projective space, then the matrix of relations will take the form

\begin{eqnarray}\label{}
Q = \begin{pmatrix}
A_{11} & 0      & \cdots & 0 & 1 \\
0      & A_{22} & \cdots & 0 & 1 \\
\vdots & \vdots & \ddots & \vdots & \vdots\\
0      & 0      & \cdots & A_{mm} & 1
\end{pmatrix}
\end{eqnarray}

\noindent
then

\begin{eqnarray}\label{lattice_proj_before_transition}
G_{\text{adm}}^{\text{max}}/J_A \cong \bigoplus_{i=1}^m \mathbb{Z}v^i \Big/ L_{Q}
\end{eqnarray}

\noindent
Since any column vectors in $R$ and $Q$ is a solution of (\ref{Diophantine_equation_0}) then we can find the coordinates of the column vectors of $R$ and $Q$ in the basis of the fundamental system of solutions (\ref{fundamental_system}), that is, the matrices $R_v$ and $Q_v$

\begin{eqnarray}\label{}
\sum_{i=1}^m v_{j}{}^{i}(R_v)_{ir} = R_{jr},\quad \sum_{i=1}^m v_{j}{}^{i}(Q_v)_{ir} = Q_{jr},
\end{eqnarray}

\noindent
where the first and second indices of the matrices indicate the row and column, respectively. Then relations (\ref{lattice_before_transition}) and (\ref{lattice_proj_before_transition}) take the form
\begin{eqnarray}\label{Ferma_lattice_modules}
G_{\text{adm}}^{\text{max}} \cong \mathbb{Z}^m \Big/ L_{R_v},\quad G_{\text{adm}}^{\text{max}}/J_A \cong \mathbb{Z}^m \Big/ L_{Q_v}.
\end{eqnarray}

\noindent
as $m = \operatorname{rk}(R_v) = \operatorname{rk}(L_{R_v})$ and $m=\operatorname{rk}(Q_v)=\operatorname{rk}(L_{Q_v})$ then the ranks of the free summands in the direct sum decomposition of the groups into cyclic groups are $\operatorname{rk}\!\left(
G_{\mathrm{adm}}^{\max}/\operatorname{Tors}(G_{\mathrm{adm}}^{\max})\right)= m - \operatorname{rk}(R_v) = 0$ and $\operatorname{rk}\!\left(
\left( G_{\mathrm{adm}}^{\max}/J_A\right)/\operatorname{Tors}\!\left(G_{\mathrm{adm}}^{\max}/J_A\right)\right)= m - \operatorname{rk}(Q_v) = 0$. Hence, there are no free summands in the decomposition of $G_{\text{adm}}^{\text{max}}$ and $G_{\text{adm}}^{\text{max}}/J_A $.

We can now apply the method described in Chapter 4 for decomposing the \(\mathbb{Z}\)-modules in (\ref{Ferma_lattice_modules}) into invariant factors
\begin{eqnarray}\label{Smith1}
&&G_{\text{adm}}^{\text{max}} \cong \bigoplus_{i=1}^m \mathbb{Z}/(\mu_i),\qquad\,\,\, \mu_l\mid\mu_p\,\,\text{for}\,\, l<p\\ 
\label{Smith2}
&&G_{\text{adm}}^{\text{max}}/J_A \cong \bigoplus_{j=1}^m \mathbb{Z}/(\nu_j),\quad  \nu_s\mid\nu_c\,\,\text{for}\,\, s<c,
\end{eqnarray}

\noindent
where $\mu_i=(D_{R_v})_{ii}$, $\nu_j=(D_{Q_v})_{jj}$ ($(D_{R_v})_{ii}\neq 0$, $(D_{Q_v})_{jj}\neq 0$) and  some of the direct summands in the decompositions may be trivial.

The approach developed in this chapter was used to obtain decompositions for all Fermat-type Calabi–Yau threefolds. The results are presented in the Appendix.

\section{Conclusion}
\hspace*{\parindent}

\vspace{-1em}

In this paper we have developed the modified algorithm for constructing all Calabi–Yau orbifolds of Berglund–Hübsch type, together with their mirror partners, and have established the correctness of the proposed construction. One of the main advantages of this approach is that it provides an explicit description of generators for the groups defining the corresponding orbifolds. Furthermore, the data produced by the algorithm can be used to determine various topological invariants of the orbifolds and their mirrors, including their Euler characteristics, Hodge numbers, and other related invariants. Thus, the proposed algorithm provides an effective framework for the systematic construction and further study of Berglund–Hübsch-type Calabi–Yau orbifolds and their mirror partners.

Furthermore, we have proposed an algorithm for obtaining all representatives of the classes in the quotient group $H/J_A$ for an arbitrary subgroup $H \subset G_{\mathrm{adm}}^{\max}$.

Another result of this work is the method for representing $G_{\mathrm{adm}}^{\max}$ and $G_{\mathrm{adm}}^{\max}/J_A$ by their invariant factor decompositions (the Frobenius representations of these groups). This approach can be used to classify maximal admissible groups and their quotients by the symmetry group of a weighted projective space for Calabi–Yau varieties.  We illustrate the application of the classification approach to an arbitrary Fermat-type Calabi–Yau manifold. In addition, as an illustrative example, we classify all maximal admissible groups of Fermat-type Calabi–Yau threefolds. Moreover, this method significantly simplifies the computations involved in Step 1 of the algorithm presented in Chapter~2.

\section{Acknowledgments}
\hspace*{\parindent}
The author acknowledges A.~Belavin and G.~Koshevoy for the helpful discussions and would like to thank K.~Stepanyantz and I~Shirokov for helpful advice. The research was carried out within the state assignment of Ministry of Science and Higher Education of the Russian Federation for IITP RAS.

\section{Appendix}

\nopagebreak

\begin{table}[H]
\centering
\begin{tabular}{|c|c|c|c|c|}
\hline
\textbf{\textnumero} &
\textbf{$\mathbf{A}$} &
\textbf{$\mathbf{k}$} &
\textbf{$\mathbf{G_{\mathrm{adm}}^{\max}/J_A}$} &
\textbf{$\mathbf{G_{\mathrm{adm}}^{\max}}$}    \\ \hline 1 & $(2, 3, 7, 43, 1806)$ & $(903, 602, 258, 42, 1)$ & $\{0\}$ & $\mathbb{Z}_{1806}$\\ \hline 2 & $(2, 3, 7, 44, 924)$ & $(462, 308, 132, 21, 1)$ & $\mathbb{Z}_{2}$ & $\mathbb{Z}_{2}\oplus\mathbb{Z}_{924}$\\ \hline 3 & $(2, 3, 7, 45, 630)$ & $(315, 210, 90, 14, 1)$ & $\mathbb{Z}_{3}$ & $\mathbb{Z}_{3}\oplus\mathbb{Z}_{630}$\\ \hline 4 & $(2, 3, 7, 46, 483)$ & $(483, 322, 138, 21, 2)$ & $\{0\}$ & $\mathbb{Z}_{966}$\\ \hline 5 & $(2, 3, 7, 48, 336)$ & $(168, 112, 48, 7, 1)$ & $\mathbb{Z}_{6}$ & $\mathbb{Z}_{6}\oplus\mathbb{Z}_{336}$\\ \hline 6 & $(2, 3, 7, 49, 294)$ & $(147, 98, 42, 6, 1)$ & $\mathbb{Z}_{7}$ & $\mathbb{Z}_{7}\oplus\mathbb{Z}_{294}$\\ \hline 7 & $(2, 3, 7, 51, 238)$ & $(357, 238, 102, 14, 3)$ & $\{0\}$ & $\mathbb{Z}_{714}$\\ \hline 8 & $(2, 3, 7, 54, 189)$ & $(189, 126, 54, 7, 2)$ & $\mathbb{Z}_{3}$ & $\mathbb{Z}_{3}\oplus\mathbb{Z}_{378}$\\ \hline 9 & $(2, 3, 7, 56, 168)$ & $(84, 56, 24, 3, 1)$ & $\mathbb{Z}_{14}$ & $\mathbb{Z}_{14}\oplus\mathbb{Z}_{168}$\\ \hline 10 & $(2, 3, 7, 60, 140)$ & $(210, 140, 60, 7, 3)$ & $\mathbb{Z}_{2}$ & $\mathbb{Z}_{2}\oplus\mathbb{Z}_{420}$\\ \hline 11 & $(2, 3, 7, 63, 126)$ & $(63, 42, 18, 2, 1)$ & $\mathbb{Z}_{21}$ & $\mathbb{Z}_{21}\oplus\mathbb{Z}_{126}$\\ \hline 12 & $(2, 3, 7, 70, 105)$ & $(105, 70, 30, 3, 2)$ & $\mathbb{Z}_{7}$ & $\mathbb{Z}_{7}\oplus\mathbb{Z}_{210}$\\ \hline 13 & $(2, 3, 7, 78, 91)$ & $(273, 182, 78, 7, 6)$ & $\{0\}$ & $\mathbb{Z}_{546}$\\ \hline 14 & $(2, 3, 7, 84, 84)$ & $(42, 28, 12, 1, 1)$ & $\mathbb{Z}_{42}$ & $\mathbb{Z}_{42}\oplus\mathbb{Z}_{84}$\\ \hline 15 & $(2, 3, 8, 25, 600)$ & $(300, 200, 75, 24, 1)$ & $\mathbb{Z}_{2}$ & $\mathbb{Z}_{2}\oplus\mathbb{Z}_{600}$\\ \hline 16 & $(2, 3, 8, 26, 312)$ & $(156, 104, 39, 12, 1)$ & $\mathbb{Z}_{2}\oplus\mathbb{Z}_{2}$ & $\mathbb{Z}_{2}\oplus\mathbb{Z}_{2}\oplus\mathbb{Z}_{312}$\\ \hline 17 & $(2, 3, 8, 27, 216)$ & $(108, 72, 27, 8, 1)$ & $\mathbb{Z}_{6}$ & $\mathbb{Z}_{6}\oplus\mathbb{Z}_{216}$\\ \hline 18 & $(2, 3, 8, 28, 168)$ & $(84, 56, 21, 6, 1)$ & $\mathbb{Z}_{2}\oplus\mathbb{Z}_{4}$ & $\mathbb{Z}_{2}\oplus\mathbb{Z}_{4}\oplus\mathbb{Z}_{168}$\\ \hline 19 & $(2, 3, 8, 30, 120)$ & $(60, 40, 15, 4, 1)$ & $\mathbb{Z}_{2}\oplus\mathbb{Z}_{6}$ & $\mathbb{Z}_{2}\oplus\mathbb{Z}_{6}\oplus\mathbb{Z}_{120}$\\ \hline 20 & $(2, 3, 8, 32, 96)$ & $(48, 32, 12, 3, 1)$ & $\mathbb{Z}_{2}\oplus\mathbb{Z}_{8}$ & $\mathbb{Z}_{2}\oplus\mathbb{Z}_{8}\oplus\mathbb{Z}_{96}$\\ \hline 21 & $(2, 3, 8, 33, 88)$ & $(132, 88, 33, 8, 3)$ & $\mathbb{Z}_{2}$ & $\mathbb{Z}_{2}\oplus\mathbb{Z}_{264}$\\ \hline 22 & $(2, 3, 8, 36, 72)$ & $(36, 24, 9, 2, 1)$ & $\mathbb{Z}_{2}\oplus\mathbb{Z}_{12}$ & $\mathbb{Z}_{2}\oplus\mathbb{Z}_{12}\oplus\mathbb{Z}_{72}$\\ \hline 23 & $(2, 3, 8, 40, 60)$ & $(60, 40, 15, 3, 2)$ & $\mathbb{Z}_{2}\oplus\mathbb{Z}_{4}$ & $\mathbb{Z}_{2}\oplus\mathbb{Z}_{4}\oplus\mathbb{Z}_{120}$\\ \hline 24 & $(2, 3, 8, 42, 56)$ & $(84, 56, 21, 4, 3)$ & $\mathbb{Z}_{2}\oplus\mathbb{Z}_{2}$ & $\mathbb{Z}_{2}\oplus\mathbb{Z}_{2}\oplus\mathbb{Z}_{168}$\\ \hline 25 & $(2, 3, 8, 48, 48)$ & $(24, 16, 6, 1, 1)$ & $\mathbb{Z}_{2}\oplus\mathbb{Z}_{24}$ & $\mathbb{Z}_{2}\oplus\mathbb{Z}_{24}\oplus\mathbb{Z}_{48}$\\ \hline 26 & $(2, 3, 9, 19, 342)$ & $(171, 114, 38, 18, 1)$ & $\mathbb{Z}_{3}$ & $\mathbb{Z}_{3}\oplus\mathbb{Z}_{342}$\\ \hline 27 & $(2, 3, 9, 20, 180)$ & $(90, 60, 20, 9, 1)$ & $\mathbb{Z}_{6}$ & $\mathbb{Z}_{6}\oplus\mathbb{Z}_{180}$\\ \hline 28 & $(2, 3, 9, 21, 126)$ & $(63, 42, 14, 6, 1)$ & $\mathbb{Z}_{3}\oplus\mathbb{Z}_{3}$ & $\mathbb{Z}_{3}\oplus\mathbb{Z}_{3}\oplus\mathbb{Z}_{126}$\\ \hline 29 & $(2, 3, 9, 22, 99)$ & $(99, 66, 22, 9, 2)$ & $\mathbb{Z}_{3}$ & $\mathbb{Z}_{3}\oplus\mathbb{Z}_{198}$\\ \hline 30 & $(2, 3, 9, 24, 72)$ & $(36, 24, 8, 3, 1)$ & $\mathbb{Z}_{3}\oplus\mathbb{Z}_{6}$ & $\mathbb{Z}_{3}\oplus\mathbb{Z}_{6}\oplus\mathbb{Z}_{72}$\\ \hline 31 & $(2, 3, 9, 27, 54)$ & $(27, 18, 6, 2, 1)$ & $\mathbb{Z}_{3}\oplus\mathbb{Z}_{9}$ & $\mathbb{Z}_{3}\oplus\mathbb{Z}_{9}\oplus\mathbb{Z}_{54}$\\ \hline 32 & $(2, 3, 9, 30, 45)$ & $(45, 30, 10, 3, 2)$ & $\mathbb{Z}_{3}\oplus\mathbb{Z}_{3}$ & $\mathbb{Z}_{3}\oplus\mathbb{Z}_{3}\oplus\mathbb{Z}_{90}$\\ \hline 33 & $(2, 3, 9, 36, 36)$ & $(18, 12, 4, 1, 1)$ & $\mathbb{Z}_{3}\oplus\mathbb{Z}_{18}$ & $\mathbb{Z}_{3}\oplus\mathbb{Z}_{18}\oplus\mathbb{Z}_{36}$\\ \hline 34 & $(2, 3, 10, 16, 240)$ & $(120, 80, 24, 15, 1)$ & $\mathbb{Z}_{2}\oplus\mathbb{Z}_{2}$ & $\mathbb{Z}_{2}\oplus\mathbb{Z}_{2}\oplus\mathbb{Z}_{240}$\\ \hline 35 & $(2, 3, 10, 18, 90)$ & $(45, 30, 9, 5, 1)$ & $\mathbb{Z}_{2}\oplus\mathbb{Z}_{6}$ & $\mathbb{Z}_{2}\oplus\mathbb{Z}_{6}\oplus\mathbb{Z}_{90}$\\ \hline 36 & $(2, 3, 10, 20, 60)$ & $(30, 20, 6, 3, 1)$ & $\mathbb{Z}_{2}\oplus\mathbb{Z}_{10}$ & $\mathbb{Z}_{2}\oplus\mathbb{Z}_{10}\oplus\mathbb{Z}_{60}$\\ \hline 37 & $(2, 3, 10, 24, 40)$ & $(60, 40, 12, 5, 3)$ & $\mathbb{Z}_{2}\oplus\mathbb{Z}_{2}$ & $\mathbb{Z}_{2}\oplus\mathbb{Z}_{2}\oplus\mathbb{Z}_{120}$\\ \hline 38 & $(2, 3, 10, 30, 30)$ & $(15, 10, 3, 1, 1)$ & $\mathbb{Z}_{2}\oplus\mathbb{Z}_{30}$ & $\mathbb{Z}_{2}\oplus\mathbb{Z}_{30}\oplus\mathbb{Z}_{30}$\\ \hline 39 & $(2, 3, 11, 14, 231)$ & $(231, 154, 42, 33, 2)$ & $\{0\}$ & $\mathbb{Z}_{462}$\\ \hline 40 & $(2, 3, 11, 15, 110)$ & $(165, 110, 30, 22, 3)$ & $\{0\}$ & $\mathbb{Z}_{330}$\\ \hline 41 & $(2, 3, 11, 22, 33)$ & $(33, 22, 6, 3, 2)$ & $\mathbb{Z}_{11}$ & $\mathbb{Z}_{11}\oplus\mathbb{Z}_{66}$\\ \hline 42 & $(2, 3, 12, 13, 156)$ & $(78, 52, 13, 12, 1)$ & $\mathbb{Z}_{6}$ & $\mathbb{Z}_{6}\oplus\mathbb{Z}_{156}$\\ \hline 43 & $(2, 3, 12, 14, 84)$ & $(42, 28, 7, 6, 1)$ & $\mathbb{Z}_{2}\oplus\mathbb{Z}_{6}$ & $\mathbb{Z}_{2}\oplus\mathbb{Z}_{6}\oplus\mathbb{Z}_{84}$\\ \hline 44 & $(2, 3, 12, 15, 60)$ & $(30, 20, 5, 4, 1)$ & $\mathbb{Z}_{3}\oplus\mathbb{Z}_{6}$ & $\mathbb{Z}_{3}\oplus\mathbb{Z}_{6}\oplus\mathbb{Z}_{60}$\\ \hline 45 & $(2, 3, 12, 16, 48)$ & $(24, 16, 4, 3, 1)$ & $\mathbb{Z}_{2}\oplus\mathbb{Z}_{12}$ & $\mathbb{Z}_{2}\oplus\mathbb{Z}_{12}\oplus\mathbb{Z}_{48}$\\ \hline 46 & $(2, 3, 12, 18, 36)$ & $(18, 12, 3, 2, 1)$ & $\mathbb{Z}_{6}\oplus\mathbb{Z}_{6}$ & $\mathbb{Z}_{6}\oplus\mathbb{Z}_{6}\oplus\mathbb{Z}_{36}$ \\ \hline   \end{tabular}  \end{table}

\begin{table}[H]
 \centering
\begin{tabular}{|c|c|c|c|c|}
\hline
\textbf{\textnumero} &
\textbf{$\mathbf{A}$} &
\textbf{$\mathbf{k}$} &
\textbf{$\mathbf{G_{\mathrm{adm}}^{\max}/J_A}$} &
\textbf{$\mathbf{G_{\mathrm{adm}}^{\max}}$}    \\ \hline 47 & $(2, 3, 12, 20, 30)$ & $(30, 20, 5, 3, 2)$ & $\mathbb{Z}_{2}\oplus\mathbb{Z}_{6}$ & $\mathbb{Z}_{2}\oplus\mathbb{Z}_{6}\oplus\mathbb{Z}_{60}$\\ \hline 48 & $(2, 3, 12, 21, 28)$ & $(42, 28, 7, 4, 3)$ & $\mathbb{Z}_{6}$ & $\mathbb{Z}_{6}\oplus\mathbb{Z}_{84}$\\ \hline 49 & $(2, 3, 12, 24, 24)$ & $(12, 8, 2, 1, 1)$ & $\mathbb{Z}_{6}\oplus\mathbb{Z}_{12}$ & $\mathbb{Z}_{6}\oplus\mathbb{Z}_{12}\oplus\mathbb{Z}_{24}$\\ \hline 50 & $(2, 3, 13, 13, 78)$ & $(39, 26, 6, 6, 1)$ & $\mathbb{Z}_{13}$ & $\mathbb{Z}_{13}\oplus\mathbb{Z}_{78}$\\ \hline 51 & $(2, 3, 14, 14, 42)$ & $(21, 14, 3, 3, 1)$ & $\mathbb{Z}_{2}\oplus\mathbb{Z}_{14}$ & $\mathbb{Z}_{2}\oplus\mathbb{Z}_{14}\oplus\mathbb{Z}_{42}$\\ \hline 52 & $(2, 3, 14, 15, 35)$ & $(105, 70, 15, 14, 6)$ & $\{0\}$ & $\mathbb{Z}_{210}$\\ \hline 53 & $(2, 3, 14, 21, 21)$ & $(21, 14, 3, 2, 2)$ & $\mathbb{Z}_{21}$ & $\mathbb{Z}_{21}\oplus\mathbb{Z}_{42}$\\ \hline 54 & $(2, 3, 15, 15, 30)$ & $(15, 10, 2, 2, 1)$ & $\mathbb{Z}_{3}\oplus\mathbb{Z}_{15}$ & $\mathbb{Z}_{3}\oplus\mathbb{Z}_{15}\oplus\mathbb{Z}_{30}$\\ \hline 55 & $(2, 3, 15, 20, 20)$ & $(30, 20, 4, 3, 3)$ & $\mathbb{Z}_{10}$ & $\mathbb{Z}_{10}\oplus\mathbb{Z}_{60}$\\ \hline 56 & $(2, 3, 16, 16, 24)$ & $(24, 16, 3, 3, 2)$ & $\mathbb{Z}_{2}\oplus\mathbb{Z}_{8}$ & $\mathbb{Z}_{2}\oplus\mathbb{Z}_{8}\oplus\mathbb{Z}_{48}$\\ \hline 57 & $(2, 3, 18, 18, 18)$ & $(9, 6, 1, 1, 1)$ & $\mathbb{Z}_{6}\oplus\mathbb{Z}_{18}$ & $\mathbb{Z}_{6}\oplus\mathbb{Z}_{18}\oplus\mathbb{Z}_{18}$\\ \hline 58 & $(2, 4, 5, 21, 420)$ & $(210, 105, 84, 20, 1)$ & $\mathbb{Z}_{2}$ & $\mathbb{Z}_{2}\oplus\mathbb{Z}_{420}$\\ \hline 59 & $(2, 4, 5, 22, 220)$ & $(110, 55, 44, 10, 1)$ & $\mathbb{Z}_{2}\oplus\mathbb{Z}_{2}$ & $\mathbb{Z}_{2}\oplus\mathbb{Z}_{2}\oplus\mathbb{Z}_{220}$\\ \hline 60 & $(2, 4, 5, 24, 120)$ & $(60, 30, 24, 5, 1)$ & $\mathbb{Z}_{2}\oplus\mathbb{Z}_{4}$ & $\mathbb{Z}_{2}\oplus\mathbb{Z}_{4}\oplus\mathbb{Z}_{120}$\\ \hline 61 & $(2, 4, 5, 25, 100)$ & $(50, 25, 20, 4, 1)$ & $\mathbb{Z}_{10}$ & $\mathbb{Z}_{10}\oplus\mathbb{Z}_{100}$\\ \hline 62 & $(2, 4, 5, 28, 70)$ & $(70, 35, 28, 5, 2)$ & $\mathbb{Z}_{2}\oplus\mathbb{Z}_{2}$ & $\mathbb{Z}_{2}\oplus\mathbb{Z}_{2}\oplus\mathbb{Z}_{140}$\\ \hline 63 & $(2, 4, 5, 30, 60)$ & $(30, 15, 12, 2, 1)$ & $\mathbb{Z}_{2}\oplus\mathbb{Z}_{10}$ & $\mathbb{Z}_{2}\oplus\mathbb{Z}_{10}\oplus\mathbb{Z}_{60}$\\ \hline 64 & $(2, 4, 5, 36, 45)$ & $(90, 45, 36, 5, 4)$ & $\mathbb{Z}_{2}$ & $\mathbb{Z}_{2}\oplus\mathbb{Z}_{180}$\\ \hline 65 & $(2, 4, 5, 40, 40)$ & $(20, 10, 8, 1, 1)$ & $\mathbb{Z}_{2}\oplus\mathbb{Z}_{20}$ & $\mathbb{Z}_{2}\oplus\mathbb{Z}_{20}\oplus\mathbb{Z}_{40}$\\ \hline 66 & $(2, 4, 6, 13, 156)$ & $(78, 39, 26, 12, 1)$ & $\mathbb{Z}_{2}\oplus\mathbb{Z}_{2}$ & $\mathbb{Z}_{2}\oplus\mathbb{Z}_{2}\oplus\mathbb{Z}_{156}$\\ \hline 67 & $(2, 4, 6, 14, 84)$ & $(42, 21, 14, 6, 1)$ & $\mathbb{Z}_{2}\oplus\mathbb{Z}_{2}\oplus\mathbb{Z}_{2}$ & $\mathbb{Z}_{2}\oplus\mathbb{Z}_{2}\oplus\mathbb{Z}_{2}\oplus\mathbb{Z}_{84}$\\ \hline 68 & $(2, 4, 6, 15, 60)$ & $(30, 15, 10, 4, 1)$ & $\mathbb{Z}_{2}\oplus\mathbb{Z}_{6}$ & $\mathbb{Z}_{2}\oplus\mathbb{Z}_{6}\oplus\mathbb{Z}_{60}$\\ \hline 69 & $(2, 4, 6, 16, 48)$ & $(24, 12, 8, 3, 1)$ & $\mathbb{Z}_{2}\oplus\mathbb{Z}_{2}\oplus\mathbb{Z}_{4}$ & $\mathbb{Z}_{2}\oplus\mathbb{Z}_{2}\oplus\mathbb{Z}_{4}\oplus\mathbb{Z}_{48}$\\ \hline 70 & $(2, 4, 6, 18, 36)$ & $(18, 9, 6, 2, 1)$ & $\mathbb{Z}_{2}\oplus\mathbb{Z}_{2}\oplus\mathbb{Z}_{6}$ & $\mathbb{Z}_{2}\oplus\mathbb{Z}_{2}\oplus\mathbb{Z}_{6}\oplus\mathbb{Z}_{36}$\\ \hline 71 & $(2, 4, 6, 20, 30)$ & $(30, 15, 10, 3, 2)$ & $\mathbb{Z}_{2}\oplus\mathbb{Z}_{2}\oplus\mathbb{Z}_{2}$ & $\mathbb{Z}_{2}\oplus\mathbb{Z}_{2}\oplus\mathbb{Z}_{2}\oplus\mathbb{Z}_{60}$\\ \hline 72 & $(2, 4, 6, 21, 28)$ & $(42, 21, 14, 4, 3)$ & $\mathbb{Z}_{2}\oplus\mathbb{Z}_{2}$ & $\mathbb{Z}_{2}\oplus\mathbb{Z}_{2}\oplus\mathbb{Z}_{84}$\\ \hline 73 & $(2, 4, 6, 24, 24)$ & $(12, 6, 4, 1, 1)$ & $\mathbb{Z}_{2}\oplus\mathbb{Z}_{2}\oplus\mathbb{Z}_{12}$ & $\mathbb{Z}_{2}\oplus\mathbb{Z}_{2}\oplus\mathbb{Z}_{12}\oplus\mathbb{Z}_{24}$\\ \hline 74 & $(2, 4, 7, 10, 140)$ & $(70, 35, 20, 14, 1)$ & $\mathbb{Z}_{2}\oplus\mathbb{Z}_{2}$ & $\mathbb{Z}_{2}\oplus\mathbb{Z}_{2}\oplus\mathbb{Z}_{140}$\\ \hline 75 & $(2, 4, 7, 12, 42)$ & $(42, 21, 12, 7, 2)$ & $\mathbb{Z}_{2}\oplus\mathbb{Z}_{2}$ & $\mathbb{Z}_{2}\oplus\mathbb{Z}_{2}\oplus\mathbb{Z}_{84}$\\ \hline 76 & $(2, 4, 7, 14, 28)$ & $(14, 7, 4, 2, 1)$ & $\mathbb{Z}_{2}\oplus\mathbb{Z}_{14}$ & $\mathbb{Z}_{2}\oplus\mathbb{Z}_{14}\oplus\mathbb{Z}_{28}$\\ \hline 77 & $(2, 4, 8, 9, 72)$ & $(36, 18, 9, 8, 1)$ & $\mathbb{Z}_{2}\oplus\mathbb{Z}_{4}$ & $\mathbb{Z}_{2}\oplus\mathbb{Z}_{4}\oplus\mathbb{Z}_{72}$\\ \hline 78 & $(2, 4, 8, 10, 40)$ & $(20, 10, 5, 4, 1)$ & $\mathbb{Z}_{2}\oplus\mathbb{Z}_{2}\oplus\mathbb{Z}_{4}$ & $\mathbb{Z}_{2}\oplus\mathbb{Z}_{2}\oplus\mathbb{Z}_{4}\oplus\mathbb{Z}_{40}$\\ \hline 79 & $(2, 4, 8, 12, 24)$ & $(12, 6, 3, 2, 1)$ & $\mathbb{Z}_{2}\oplus\mathbb{Z}_{4}\oplus\mathbb{Z}_{4}$ & $\mathbb{Z}_{2}\oplus\mathbb{Z}_{4}\oplus\mathbb{Z}_{4}\oplus\mathbb{Z}_{24}$\\ \hline 80 & $(2, 4, 8, 16, 16)$ & $(8, 4, 2, 1, 1)$ & $\mathbb{Z}_{2}\oplus\mathbb{Z}_{4}\oplus\mathbb{Z}_{8}$ & $\mathbb{Z}_{2}\oplus\mathbb{Z}_{4}\oplus\mathbb{Z}_{8}\oplus\mathbb{Z}_{16}$\\ \hline 81 & $(2, 4, 9, 9, 36)$ & $(18, 9, 4, 4, 1)$ & $\mathbb{Z}_{18}$ & $\mathbb{Z}_{18}\oplus\mathbb{Z}_{36}$\\ \hline 82 & $(2, 4, 9, 12, 18)$ & $(18, 9, 4, 3, 2)$ & $\mathbb{Z}_{2}\oplus\mathbb{Z}_{6}$ & $\mathbb{Z}_{2}\oplus\mathbb{Z}_{6}\oplus\mathbb{Z}_{36}$\\ \hline 83 & $(2, 4, 10, 10, 20)$ & $(10, 5, 2, 2, 1)$ & $\mathbb{Z}_{2}\oplus\mathbb{Z}_{2}\oplus\mathbb{Z}_{10}$ & $\mathbb{Z}_{2}\oplus\mathbb{Z}_{2}\oplus\mathbb{Z}_{10}\oplus\mathbb{Z}_{20}$\\ \hline 84 & $(2, 4, 10, 12, 15)$ & $(30, 15, 6, 5, 4)$ & $\mathbb{Z}_{2}\oplus\mathbb{Z}_{2}$ & $\mathbb{Z}_{2}\oplus\mathbb{Z}_{2}\oplus\mathbb{Z}_{60}$\\ \hline 85 & $(2, 4, 12, 12, 12)$ & $(6, 3, 1, 1, 1)$ & $\mathbb{Z}_{2}\oplus\mathbb{Z}_{4}\oplus\mathbb{Z}_{12}$ & $\mathbb{Z}_{2}\oplus\mathbb{Z}_{4}\oplus\mathbb{Z}_{12}\oplus\mathbb{Z}_{12}$\\ \hline 86 & $(2, 5, 5, 11, 110)$ & $(55, 22, 22, 10, 1)$ & $\mathbb{Z}_{5}$ & $\mathbb{Z}_{5}\oplus\mathbb{Z}_{110}$\\ \hline 87 & $(2, 5, 5, 12, 60)$ & $(30, 12, 12, 5, 1)$ & $\mathbb{Z}_{10}$ & $\mathbb{Z}_{10}\oplus\mathbb{Z}_{60}$\\ \hline 88 & $(2, 5, 5, 14, 35)$ & $(35, 14, 14, 5, 2)$ & $\mathbb{Z}_{5}$ & $\mathbb{Z}_{5}\oplus\mathbb{Z}_{70}$\\ \hline 89 & $(2, 5, 5, 15, 30)$ & $(15, 6, 6, 2, 1)$ & $\mathbb{Z}_{5}\oplus\mathbb{Z}_{5}$ & $\mathbb{Z}_{5}\oplus\mathbb{Z}_{5}\oplus\mathbb{Z}_{30}$\\ \hline 90 & $(2, 5, 5, 20, 20)$ & $(10, 4, 4, 1, 1)$ & $\mathbb{Z}_{5}\oplus\mathbb{Z}_{10}$ & $\mathbb{Z}_{5}\oplus\mathbb{Z}_{10}\oplus\mathbb{Z}_{20}$\\ \hline 91 & $(2, 5, 6, 8, 120)$ & $(60, 24, 20, 15, 1)$ & $\mathbb{Z}_{2}\oplus\mathbb{Z}_{2}$ & $\mathbb{Z}_{2}\oplus\mathbb{Z}_{2}\oplus\mathbb{Z}_{120}$\\ \hline 92 & $(2, 5, 6, 9, 45)$ & $(45, 18, 15, 10, 2)$ & $\mathbb{Z}_{3}$ & $\mathbb{Z}_{3}\oplus\mathbb{Z}_{90}$\\ \hline 93 & $(2, 5, 6, 10, 30)$ & $(15, 6, 5, 3, 1)$ & $\mathbb{Z}_{2}\oplus\mathbb{Z}_{10}$ & $\mathbb{Z}_{2}\oplus\mathbb{Z}_{10}\oplus\mathbb{Z}_{30}$ \\ \hline   \end{tabular}  \end{table}

\begin{table}[H]
\centering
\begin{tabular}{|c|c|c|c|c|}
\hline
\textbf{\textnumero} &
\textbf{$\mathbf{A}$} &
\textbf{$\mathbf{k}$} &
\textbf{$\mathbf{G_{\mathrm{adm}}^{\max}/J_A}$} &
\textbf{$\mathbf{G_{\mathrm{adm}}^{\max}}$}    \\ \hline 94 & $(2, 5, 6, 12, 20)$ & $(30, 12, 10, 5, 3)$ & $\mathbb{Z}_{2}\oplus\mathbb{Z}_{2}$ & $\mathbb{Z}_{2}\oplus\mathbb{Z}_{2}\oplus\mathbb{Z}_{60}$\\ \hline 95 & $(2, 5, 6, 15, 15)$ & $(15, 6, 5, 2, 2)$ & $\mathbb{Z}_{15}$ & $\mathbb{Z}_{15}\oplus\mathbb{Z}_{30}$\\ \hline 96 & $(2, 5, 7, 7, 70)$ & $(35, 14, 10, 10, 1)$ & $\mathbb{Z}_{7}$ & $\mathbb{Z}_{7}\oplus\mathbb{Z}_{70}$\\ \hline 97 & $(2, 5, 8, 8, 20)$ & $(20, 8, 5, 5, 2)$ & $\mathbb{Z}_{2}\oplus\mathbb{Z}_{4}$ & $\mathbb{Z}_{2}\oplus\mathbb{Z}_{4}\oplus\mathbb{Z}_{40}$\\ \hline 98 & $(2, 5, 10, 10, 10)$ & $(5, 2, 1, 1, 1)$ & $\mathbb{Z}_{10}\oplus\mathbb{Z}_{10}$ & $\mathbb{Z}_{10}\oplus\mathbb{Z}_{10}\oplus\mathbb{Z}_{10}$\\ \hline 99 & $(2, 6, 6, 7, 42)$ & $(21, 7, 7, 6, 1)$ & $\mathbb{Z}_{2}\oplus\mathbb{Z}_{6}$ & $\mathbb{Z}_{2}\oplus\mathbb{Z}_{6}\oplus\mathbb{Z}_{42}$\\ \hline 100 & $(2, 6, 6, 8, 24)$ & $(12, 4, 4, 3, 1)$ & $\mathbb{Z}_{2}\oplus\mathbb{Z}_{2}\oplus\mathbb{Z}_{6}$ & $\mathbb{Z}_{2}\oplus\mathbb{Z}_{2}\oplus\mathbb{Z}_{6}\oplus\mathbb{Z}_{24}$\\ \hline 101 & $(2, 6, 6, 9, 18)$ & $(9, 3, 3, 2, 1)$ & $\mathbb{Z}_{6}\oplus\mathbb{Z}_{6}$ & $\mathbb{Z}_{6}\oplus\mathbb{Z}_{6}\oplus\mathbb{Z}_{18}$\\ \hline 102 & $(2, 6, 6, 10, 15)$ & $(15, 5, 5, 3, 2)$ & $\mathbb{Z}_{2}\oplus\mathbb{Z}_{6}$ & $\mathbb{Z}_{2}\oplus\mathbb{Z}_{6}\oplus\mathbb{Z}_{30}$\\ \hline 103 & $(2, 6, 6, 12, 12)$ & $(6, 2, 2, 1, 1)$ & $\mathbb{Z}_{2}\oplus\mathbb{Z}_{6}\oplus\mathbb{Z}_{6}$ & $\mathbb{Z}_{2}\oplus\mathbb{Z}_{6}\oplus\mathbb{Z}_{6}\oplus\mathbb{Z}_{12}$\\ \hline 104 & $(2, 6, 7, 7, 21)$ & $(21, 7, 6, 6, 2)$ & $\mathbb{Z}_{7}$ & $\mathbb{Z}_{7}\oplus\mathbb{Z}_{42}$\\ \hline 105 & $(2, 6, 8, 8, 12)$ & $(12, 4, 3, 3, 2)$ & $\mathbb{Z}_{2}\oplus\mathbb{Z}_{2}\oplus\mathbb{Z}_{4}$ & $\mathbb{Z}_{2}\oplus\mathbb{Z}_{2}\oplus\mathbb{Z}_{4}\oplus\mathbb{Z}_{24}$\\ \hline 106 & $(2, 6, 9, 9, 9)$ & $(9, 3, 2, 2, 2)$ & $\mathbb{Z}_{3}\oplus\mathbb{Z}_{9}$ & $\mathbb{Z}_{3}\oplus\mathbb{Z}_{9}\oplus\mathbb{Z}_{18}$\\ \hline 107 & $(2, 7, 7, 7, 14)$ & $(7, 2, 2, 2, 1)$ & $\mathbb{Z}_{7}\oplus\mathbb{Z}_{7}$ & $\mathbb{Z}_{7}\oplus\mathbb{Z}_{7}\oplus\mathbb{Z}_{14}$\\ \hline 108 & $(2, 8, 8, 8, 8)$ & $(4, 1, 1, 1, 1)$ & $\mathbb{Z}_{2}\oplus\mathbb{Z}_{8}\oplus\mathbb{Z}_{8}$ & $\mathbb{Z}_{2}\oplus\mathbb{Z}_{8}\oplus\mathbb{Z}_{8}\oplus\mathbb{Z}_{8}$\\ \hline 109 & $(3, 3, 4, 13, 156)$ & $(52, 52, 39, 12, 1)$ & $\mathbb{Z}_{3}$ & $\mathbb{Z}_{3}\oplus\mathbb{Z}_{156}$\\ \hline 110 & $(3, 3, 4, 14, 84)$ & $(28, 28, 21, 6, 1)$ & $\mathbb{Z}_{6}$ & $\mathbb{Z}_{6}\oplus\mathbb{Z}_{84}$\\ \hline 111 & $(3, 3, 4, 15, 60)$ & $(20, 20, 15, 4, 1)$ & $\mathbb{Z}_{3}\oplus\mathbb{Z}_{3}$ & $\mathbb{Z}_{3}\oplus\mathbb{Z}_{3}\oplus\mathbb{Z}_{60}$\\ \hline 112 & $(3, 3, 4, 16, 48)$ & $(16, 16, 12, 3, 1)$ & $\mathbb{Z}_{12}$ & $\mathbb{Z}_{12}\oplus\mathbb{Z}_{48}$\\ \hline 113 & $(3, 3, 4, 18, 36)$ & $(12, 12, 9, 2, 1)$ & $\mathbb{Z}_{3}\oplus\mathbb{Z}_{6}$ & $\mathbb{Z}_{3}\oplus\mathbb{Z}_{6}\oplus\mathbb{Z}_{36}$\\ \hline 114 & $(3, 3, 4, 20, 30)$ & $(20, 20, 15, 3, 2)$ & $\mathbb{Z}_{6}$ & $\mathbb{Z}_{6}\oplus\mathbb{Z}_{60}$\\ \hline 115 & $(3, 3, 4, 21, 28)$ & $(28, 28, 21, 4, 3)$ & $\mathbb{Z}_{3}$ & $\mathbb{Z}_{3}\oplus\mathbb{Z}_{84}$\\ \hline 116 & $(3, 3, 4, 24, 24)$ & $(8, 8, 6, 1, 1)$ & $\mathbb{Z}_{3}\oplus\mathbb{Z}_{12}$ & $\mathbb{Z}_{3}\oplus\mathbb{Z}_{12}\oplus\mathbb{Z}_{24}$\\ \hline 117 & $(3, 3, 5, 8, 120)$ & $(40, 40, 24, 15, 1)$ & $\mathbb{Z}_{3}$ & $\mathbb{Z}_{3}\oplus\mathbb{Z}_{120}$\\ \hline 118 & $(3, 3, 5, 9, 45)$ & $(15, 15, 9, 5, 1)$ & $\mathbb{Z}_{3}\oplus\mathbb{Z}_{3}$ & $\mathbb{Z}_{3}\oplus\mathbb{Z}_{3}\oplus\mathbb{Z}_{45}$\\ \hline 119 & $(3, 3, 5, 10, 30)$ & $(10, 10, 6, 3, 1)$ & $\mathbb{Z}_{15}$ & $\mathbb{Z}_{15}\oplus\mathbb{Z}_{30}$\\ \hline 120 & $(3, 3, 5, 12, 20)$ & $(20, 20, 12, 5, 3)$ & $\mathbb{Z}_{3}$ & $\mathbb{Z}_{3}\oplus\mathbb{Z}_{60}$\\ \hline 121 & $(3, 3, 5, 15, 15)$ & $(5, 5, 3, 1, 1)$ & $\mathbb{Z}_{3}\oplus\mathbb{Z}_{15}$ & $\mathbb{Z}_{3}\oplus\mathbb{Z}_{15}\oplus\mathbb{Z}_{15}$\\ \hline 122 & $(3, 3, 6, 7, 42)$ & $(14, 14, 7, 6, 1)$ & $\mathbb{Z}_{3}\oplus\mathbb{Z}_{3}$ & $\mathbb{Z}_{3}\oplus\mathbb{Z}_{3}\oplus\mathbb{Z}_{42}$\\ \hline 123 & $(3, 3, 6, 8, 24)$ & $(8, 8, 4, 3, 1)$ & $\mathbb{Z}_{3}\oplus\mathbb{Z}_{6}$ & $\mathbb{Z}_{3}\oplus\mathbb{Z}_{6}\oplus\mathbb{Z}_{24}$\\ \hline 124 & $(3, 3, 6, 9, 18)$ & $(6, 6, 3, 2, 1)$ & $\mathbb{Z}_{3}\oplus\mathbb{Z}_{3}\oplus\mathbb{Z}_{3}$ & $\mathbb{Z}_{3}\oplus\mathbb{Z}_{3}\oplus\mathbb{Z}_{3}\oplus\mathbb{Z}_{18}$\\ \hline 125 & $(3, 3, 6, 10, 15)$ & $(10, 10, 5, 3, 2)$ & $\mathbb{Z}_{3}\oplus\mathbb{Z}_{3}$ & $\mathbb{Z}_{3}\oplus\mathbb{Z}_{3}\oplus\mathbb{Z}_{30}$\\ \hline 126 & $(3, 3, 6, 12, 12)$ & $(4, 4, 2, 1, 1)$ & $\mathbb{Z}_{3}\oplus\mathbb{Z}_{3}\oplus\mathbb{Z}_{6}$ & $\mathbb{Z}_{3}\oplus\mathbb{Z}_{3}\oplus\mathbb{Z}_{6}\oplus\mathbb{Z}_{12}$\\ \hline 127 & $(3, 3, 7, 7, 21)$ & $(7, 7, 3, 3, 1)$ & $\mathbb{Z}_{21}$ & $\mathbb{Z}_{21}\oplus\mathbb{Z}_{21}$\\ \hline 128 & $(3, 3, 8, 8, 12)$ & $(8, 8, 3, 3, 2)$ & $\mathbb{Z}_{12}$ & $\mathbb{Z}_{12}\oplus\mathbb{Z}_{24}$\\ \hline 129 & $(3, 3, 9, 9, 9)$ & $(3, 3, 1, 1, 1)$ & $\mathbb{Z}_{3}\oplus\mathbb{Z}_{3}\oplus\mathbb{Z}_{9}$ & $\mathbb{Z}_{3}\oplus\mathbb{Z}_{3}\oplus\mathbb{Z}_{9}\oplus\mathbb{Z}_{9}$\\ \hline 130 & $(3, 4, 4, 7, 42)$ & $(28, 21, 21, 12, 2)$ & $\mathbb{Z}_{2}$ & $\mathbb{Z}_{2}\oplus\mathbb{Z}_{84}$\\ \hline 131 & $(3, 4, 4, 8, 24)$ & $(8, 6, 6, 3, 1)$ & $\mathbb{Z}_{4}\oplus\mathbb{Z}_{4}$ & $\mathbb{Z}_{4}\oplus\mathbb{Z}_{4}\oplus\mathbb{Z}_{24}$\\ \hline 132 & $(3, 4, 4, 9, 18)$ & $(12, 9, 9, 4, 2)$ & $\mathbb{Z}_{6}$ & $\mathbb{Z}_{6}\oplus\mathbb{Z}_{36}$\\ \hline 133 & $(3, 4, 4, 10, 15)$ & $(20, 15, 15, 6, 4)$ & $\mathbb{Z}_{2}$ & $\mathbb{Z}_{2}\oplus\mathbb{Z}_{60}$\\ \hline 134 & $(3, 4, 4, 12, 12)$ & $(4, 3, 3, 1, 1)$ & $\mathbb{Z}_{4}\oplus\mathbb{Z}_{12}$ & $\mathbb{Z}_{4}\oplus\mathbb{Z}_{12}\oplus\mathbb{Z}_{12}$\\ \hline 135 & $(3, 4, 5, 5, 60)$ & $(20, 15, 12, 12, 1)$ & $\mathbb{Z}_{5}$ & $\mathbb{Z}_{5}\oplus\mathbb{Z}_{60}$\\ \hline 136 & $(3, 4, 5, 6, 20)$ & $(20, 15, 12, 10, 3)$ & $\mathbb{Z}_{2}$ & $\mathbb{Z}_{2}\oplus\mathbb{Z}_{60}$\\ \hline 137 & $(3, 4, 6, 6, 12)$ & $(4, 3, 2, 2, 1)$ & $\mathbb{Z}_{6}\oplus\mathbb{Z}_{6}$ & $\mathbb{Z}_{6}\oplus\mathbb{Z}_{6}\oplus\mathbb{Z}_{12}$\\ \hline 138 & $(3, 4, 6, 8, 8)$ & $(8, 6, 4, 3, 3)$ & $\mathbb{Z}_{2}\oplus\mathbb{Z}_{4}$ & $\mathbb{Z}_{2}\oplus\mathbb{Z}_{4}\oplus\mathbb{Z}_{24}$\\ \hline 139 & $(3, 5, 5, 5, 15)$ & $(5, 3, 3, 3, 1)$ & $\mathbb{Z}_{5}\oplus\mathbb{Z}_{5}$ & $\mathbb{Z}_{5}\oplus\mathbb{Z}_{5}\oplus\mathbb{Z}_{15}$\\ \hline 140 & $(3, 5, 5, 6, 10)$ & $(10, 6, 6, 5, 3)$ & $\mathbb{Z}_{5}$ & $\mathbb{Z}_{5}\oplus\mathbb{Z}_{30}$\\ \hline 141 & $(3, 6, 6, 6, 6)$ & $(2, 1, 1, 1, 1)$ & $\mathbb{Z}_{3}\oplus\mathbb{Z}_{6}\oplus\mathbb{Z}_{6}$ & $\mathbb{Z}_{3}\oplus\mathbb{Z}_{6}\oplus\mathbb{Z}_{6}\oplus\mathbb{Z}_{6}$\\ \hline   \end{tabular} \end{table}

\begin{table}[H]
\centering
\begin{tabular}{|c|c|c|c|c|}
\hline
\textbf{\textnumero} &
\textbf{$\mathbf{A}$} &
\textbf{$\mathbf{k}$} &
\textbf{$\mathbf{G_{\mathrm{adm}}^{\max}/J_A}$} &
\textbf{$\mathbf{G_{\mathrm{adm}}^{\max}}$}
\\ \hline
 142 & $(4, 4, 4, 5, 20)$ & $(5, 5, 5, 4, 1)$ & $\mathbb{Z}_{4}\oplus\mathbb{Z}_{4}$ & $\mathbb{Z}_{4}\oplus\mathbb{Z}_{4}\oplus\mathbb{Z}_{20}$\\ \hline 143 & $(4, 4, 4, 6, 12)$ & $(3, 3, 3, 2, 1)$ & $\mathbb{Z}_{2}\oplus\mathbb{Z}_{4}\oplus\mathbb{Z}_{4}$ & $\mathbb{Z}_{2}\oplus\mathbb{Z}_{4}\oplus\mathbb{Z}_{4}\oplus\mathbb{Z}_{12}$\\ \hline 144 & $(4, 4, 4, 8, 8)$ & $(2, 2, 2, 1, 1)$ & $\mathbb{Z}_{4}\oplus\mathbb{Z}_{4}\oplus\mathbb{Z}_{4}$ & $\mathbb{Z}_{4}\oplus\mathbb{Z}_{4}\oplus\mathbb{Z}_{4}\oplus\mathbb{Z}_{8}$\\ \hline 145 & $(4, 4, 5, 5, 10)$ & $(5, 5, 4, 4, 2)$ & $\mathbb{Z}_{10}$ & $\mathbb{Z}_{10}\oplus\mathbb{Z}_{20}$\\ \hline 146 & $(4, 4, 6, 6, 6)$ & $(3, 3, 2, 2, 2)$ & $\mathbb{Z}_{2}\oplus\mathbb{Z}_{2}\oplus\mathbb{Z}_{6}$ & $\mathbb{Z}_{2}\oplus\mathbb{Z}_{2}\oplus\mathbb{Z}_{6}\oplus\mathbb{Z}_{12}$\\ \hline 147 & $(5, 5, 5, 5, 5)$ & $(1, 1, 1, 1, 1)$ & $\mathbb{Z}_{5}\oplus\mathbb{Z}_{5}\oplus\mathbb{Z}_{5}$ & $\mathbb{Z}_{5}\oplus\mathbb{Z}_{5}\oplus\mathbb{Z}_{5}\oplus\mathbb{Z}_{5}$ \\ \hline   \end{tabular} 
 
\caption{Invariant factor decomposition of the groups $G_{\mathrm{adm}}^{\max}/J_A$ and $G_{\mathrm{adm}}^{\max}$ defining all Fermat-type Calabi–Yau threefolds.}
\label{tab:decompositions} 
  \end{table}

\end{document}